\documentclass{ieeeconf}
\usepackage[
    top    = 54pt,
    bottom = 54pt,
    left   = 54pt,
    right  = 54pt]{geometry}
\usepackage{amssymb,amsmath}
\usepackage{verbatim, color, tabularx, booktabs, graphicx}
\usepackage[font=small]{caption}
\usepackage{mathtools}
\mathtoolsset{showonlyrefs}

\newtheorem{theorem}{Theorem}
\newtheorem{definition}[theorem]{Definition}
\newtheorem{lemma}[theorem]{Lemma}

\newtheorem{remark}{Remark}
\newtheorem{proposition}[theorem]{Proposition}

\newcommand{\reals}{\mathbb{R}}											% The set of real numbers
\newcommand{\prob}{\mathbb{P}}											% The probability symbol
\newcommand{\expe}{\mathbb{E}}											% The expectation symbol
\newcommand{\todo}[1]{}													% ToDo macro
\title{ \textbf{Optimality of the Laplace Mechanism in Differential Privacy} }
\author{Fragkiskos Koufogiannis, Shuo Han, George J. Pappas 
		\thanks{Authors are with the Department of Electrical and Systems Engineering, University of Pennsylvania, PA, USA.}
		\thanks{This work was supported in part by the TerraSwarm Research Center, one of six centers supported by the STARnet phase of the Focus Center Research Program (FCRP) a Semiconductor Research Corporation program sponsored by MARCO and DARPA.} }

\begin{document}

\maketitle

\begin{abstract} In the highly interconnected realm of Internet of Things, exchange of sensitive information raises severe privacy concerns. The Laplace mechanism -- adding Laplace-distributed artificial noise to sensitive data --  is one of the widely used methods of providing privacy guarantees within the framework of differential privacy. In this work, we present Lipschitz privacy, a slightly tighter version of differential privacy. We prove that the Laplace mechanism is optimal in the sense that it minimizes the mean-squared error for identity queries which provide privacy with respect to the $\ell_{1}$-norm. In addition to the $\ell_{1}$-norm which respects individuals' participation, we focus on the use of the $\ell_{2}$-norm which provides privacy of high-dimensional data. A variation of the Laplace mechanism is proven to have the optimal mean-squared error from the identity query. Finally, the optimal mechanism for the scenario in which individuals submit their high-dimensional sensitive data is derived.
\end{abstract}

\section{Introduction}
The Internet of Things (IoT) envisions that everyday devices such as smartphones, power meters, and household appliances will exchange information and provide innovative services such as e-health and assisted living \cite{atzori10}. However, when a device communicates sensitive information (e.g. monitored activities, health records) over a vast network of interconnected \textit{things}, privacy concerns are raised \cite{weber2010internet}. For example, traffic maps can be constructed from aggregating users' GPS traces and users can benefit from such published maps by avoiding congested routes. On the other hand publishing statistics of sensitive data of a population while providing privacy guarantees is not trivial. The Netflix prize is an example were, given publicly released information \cite{bennett2007netflix}, an adversary can partially reconstruct private data \cite{narayanan06}. Accurate, privacy-preserving mechanisms are essential for IoT to provide these services while respecting individuals' privacy \cite{miorandi2012internet}

Significant efforts have been made to address these privacy concerns \cite{venkitasubramaniam2013privacy}, \cite{han2014differentially}, \cite{meingast2006security}, \cite{cardenas2012game}, \cite{weeraddana2013per}, \cite{canepa2013framework}, \cite{leny12}. Intuitively, uncertainty about the private data is introduced by publishing a perturbed response instead of the exact one. In the context of traffic monitoring, virtual trip lines and data cloaking techniques \cite{hoh2007preserving}, \cite{hoh2008virtual} provide privacy against a given adversarial model. In practice, though, an adversary may be more powerful or informed than the model assumptions.  Additionally, an information-theoretic framework based on mutual information was introduced \cite{rajagopalan2011smart}. However, this approach provides privacy guarantees in a probabilistic sense and, therefore, rare, but severe, privacy breaches are possible.

A rigorous notion of privacy is differential privacy which provides formal privacy guarantees without any assumptions on the adversary's power \cite{dwork2013algorithmic} and is the notion used in this work. Specifically, while answering queries from private data, artificial noise is injected. This noise is deliberately designed and ensures that an adversary cannot \textit{confidently infer} any individual's private data, where an \textit{adjacency relation} defines the pairs of inputs that are rendered almost indistinguishable. For tight privacy level, increased amounts of noise are required and, consequently, the accuracy of the noisy response degrades. Thus, a trade-off between privacy level and accuracy exists. Ideally, one would like to design optimal mechanisms that satisfy a predefined privacy level and approximate a given query with minimum mean-squared error.

Several methods for constructing differentially private mechanisms have been proposed. In particular, given a score function for every pair of private input and public response, the exponential mechanism \cite{mcsherry07} provides a powerful way of building a private mechanism, although no performance guarantees were initially provided. The Laplace mechanism is an instance of the exponential mechanism for real, vector-valued private data which adds Laplace-distributed noise $V$ to the private data:
\begin{align} \label{eqn:laplaceDistribution}
	\prob( V = v ) \propto e^{ -\epsilon \|v\|_{1} },
\end{align} 
where $\epsilon\in(0,\infty)$ is the privacy level --- smaller values of parameter $\epsilon$ result to stronger privacy guarantees --- and $\|\cdot\|_{1}$ is the $\ell_{1}$-norm. Near-optimality of the Laplace mechanism for a single integer-valued linear query was presented in \cite{ghosh2012universally}, whereas, for linear queries, asymptotic (in the number of users) sub-optimality bounds were derived for a variant of the Laplace mechanism \cite{hardt10}. For single-dimensional private data, the exact optimality of the ``staircase'' mechanism, a quantized version of the Laplace mechanism, was established in \cite{geng12}. Moreover, the Laplace mechanism was proven to be an entropy-minimizing private mechanism \cite{wang14} under a version of differential privacy for metric spaces \cite{chatzikokolakis13}.

In this work, we establish optimality guarantees for the Laplace mechanism -- adding Laplace-distributed noise \eqref{eqn:laplaceDistribution}. We formalize \textit{Lipschitz privacy} which is a slightly stronger version of differential privacy for metric spaces and allows us to pose the problem of designing optimal privacy-aware mechanisms as optimization problems where privacy requirements are included as constraints. We, first, prove that the Laplace mechanism optimally approximates real-valued private data by achieving the minimum \textit{mean-squared error}. Besides the $\ell_{1}$-norm used in \eqref{eqn:laplaceDistribution}, we focus on the $\ell_{2}$-norm as the appropriate adjacency relation that captures the privacy aspects of sensitive signals, such as GPS and power consumption traces. In the $\ell_{2}$-norm case, we prove the optimality of a variant of the Laplace mechanism. Furthermore, we extend our optimality results to the case of a composite adjacency relation for the scenario when multiple individuals contribute their private signals, e.g. drivers report their GPS traces. 

A brief overview of differential privacy is provided in Section \ref{sec:differentialPrivacy}. In Section \ref{sec:infinitesimalPrivacy}, a version of differential privacy for Euclidean spaces is explored and strong connections with differential privacy are established. Section \ref{sec:optimalityResults} establishes the optimal private mechanism for the case of multi-dimensional identity queries both under $\ell_{1}$ and $\ell_{2}$ norms.  We conclude this work with a discussion in Section \ref{sec:discussion}.

\section{Differential Privacy Overview} \label{sec:differentialPrivacy}
The framework of differential privacy was introduced in \cite{dwork06}, \cite{dwork2013algorithmic}. According to this framework, whenever a query is submitted to private data, the exact response must be perturbed by noise upon release to the public. Formally, the definition of differential privacy is the following:

\begin{definition} \label{def:diffPrivacy}
Let $\epsilon\geq0$ be a given privacy level, $\mathcal{U}$ be the set of possible private data, $\mathcal{A}\subseteq\mathcal{U}^{2}$ be an adjacency relation over the private data, $\mathcal{Y}$ be the set of possible responses, and $\Delta\left(\mathcal{Y}\right)$  be the set of probability measures over (a sufficiently rich $\sigma$-algebra of) $\mathcal{Y}$. A mechanism $Q:\mathcal{U} \rightarrow \Delta\left( \mathcal{Y} \right)$ is $\epsilon$-differentially private if
\begin{align}
\mathbb{P}( Qu \in \mathcal{S} ) \leq e^{\epsilon} \mathbb{P}( Qu' \in \mathcal{S} )
\end{align}
for every $S\subseteq\mathcal{Y}$ and every $u,u'\in\mathcal{U}$ such that $(u,u')\in\mathcal{A}$.
\end{definition}

%% ASSUMPTIONS ON SIGMA ALGEBRAS AND SET OF PROBABILITY MEASURES %%
\begin{remark}
For a given output set $\mathcal{Y}$, we assume the existence of a rich enough $\sigma$-algebra $M\subseteq 2^{\mathcal{Y}}$. Slightly abusing of notation, we write $\mathcal{S}\subseteq\mathcal{Y}$ instead of $\mathcal{S}\in M$. Also, the set of probability measures over $(\mathcal{Y},M)$ is denoted by $\Delta\left( \mathcal{Y} \right)$. For a finite set of responses $\mathcal{Y}$, we assume $M=2^{\mathcal{Y}}$. In this approach, we focus on Euclidean spaces $\mathcal{Y}=\mathbb{R}^{m}$ and the Borel set $M=\mathcal{B}^{m}$.
\end{remark}

Definition \ref{def:diffPrivacy} considers randomized mappings, called mechanisms, from private data in $\mathcal{U}$ to responses in $\mathcal{Y}$. The adjacency relation $\mathcal{A}$ defines the pairs of inputs $(u,u')$ that are rendered almost indistinguishable to an adversary who observes only the response of the mechanism. The level of privacy is controlled by the parameter $\epsilon\geq0$. Complete privacy is guaranteed for $\epsilon=0$, whereas, no privacy is respected for $\epsilon \rightarrow \infty$. A differential private algorithm is a map from private data to distributions over the set of responses. Upon release, the differential private response is given by a single random sample drawn from the distribution.

A differential private mechanism needs to be useful at the same time. For example, a mechanism that responds identically for any input is $0$-differential private, but also useless. To this end, we are interested in mechanisms $Q_{\epsilon}$ that approximate a given query $q:\mathcal{U}\rightarrow\mathcal{Y}$. We say that an $\epsilon$-differential private mechanism is optimal (in the mean-squared sense) if it minimizes the mean-squared error of the desired query $q$. Characterization of the optimal private mechanism is fundamental for efficient applications of differential privacy.

In this work, we present optimal private mechanisms for identity queries under a general adjacency relation. Specifically, we focus on Euclidean spaces and assume each of the $n$ users contributes his $m$-dimensional sensitive data. Let $\mathcal{U}=\mathbb{R}^{n\times m}$ and $\mathcal{Y}=\mathbb{R}^{n \times m}$ and consider the adjacency relation $\mathcal{A}$ defined as:
\begin{align} \label{eqn:hybridAdjacency}
(u,u') \in \mathcal{A} \Leftrightarrow \exists i \text{ s.t. }  \|u_{i}-u'_{i}\|_{2}\leq\alpha \text{ and } u_{j}=u'_{j}, \forall j\neq i.
\end{align}
Adjacency relation \eqref{eqn:hybridAdjacency} respects privacy of every individual's sensitive data $u_{i}$; even if an adversary is aware of every other user's data $u_{j}, j\neq i$, the adversary cannot confidently extract the value $u_{i}$.

%Despite focusing only on identity queries, the applicatios the framework of differential privacy is complemented with numerous results that include the exponential mechanism and resilience to post-processing \cite{mcsherry07} and allow more complex private mechanisms to be constructed from simpler ones.

%Given a prior $\pi\in\Delta\left(\mathcal{U}\right)$ over the inputs and a loss function $l:\mathcal{Y}^{2}\rightarrow\mathbb{R}$, the problem of designing an efficient $\epsilon$-dp mechanism can be stated as following:
%\begin{align} \begin{split}
%\underset{ Q:\mathcal{U}\rightarrow \Delta\left(\mathcal{Y}\right) }{ \text{minimize} } \quad & \underset{u\sim\pi}{\mathbb{E}} \underset{y\sim Qu}{\mathbb{E}} l( q(u), y ) \\
%\text{s.t.} \quad & \mathbb{P}(Qu\in\mathcal{S}) \leq e^{\epsilon} \mathcal{P}(Qu'\in\mathcal{S}), \: \forall u,u':\:\|u-u'\|\leq\alpha, \text{ and } \forall \mathcal{S}\subseteq\mathcal{Y}
%\end{split} \end{align}
%Let $Q_{\epsilon}$ denote the solution of the optimization problem. Except for special cases, solving for the optimal $\epsilon$-dp mechanism is an open problem. The main hardness stems from the infinite dimensionality of the solution space.

%% INTRODUCING LIPSCHTIZ PRIVACY, FORMERLY MENTIONED AS INFINITESIMAL PRIVACY
\section{Differential Privacy as Lipschitz Constraint} \label{sec:infinitesimalPrivacy}

%% FORMAL DEFINITION
In this section, we reformulate differential privacy for metric spaces as a Lipschitz constraint. This reformulation, which we call \textit{Lipschitz privacy}, is closely related to the original notion of differential privacy introduced in \cite{chatzikokolakis13}. In particular, the differential privacy constraint is viewed as a sensitivity constraint. The sensitive data is assumed to be an element of a complete vector space $\mathcal{U}$ equipped with a norm $\|\cdot\|$, and the set of possible responses is denoted by $\mathcal{Y}$. Formally, we provide the definition of Lipschitz privacy:
\begin{definition}[Lipschitz privacy]
Let $\mathcal{U}$ be a metric space and $\mathcal{Y}$ be a set of responses. A mechanism $Q:\mathcal{U}\rightarrow\Delta\left( \mathcal{Y} \right)$ is called $\epsilon$-Lipshcitz differentially private if the log-probability function is $\epsilon$-Lipschitz:
\begin{align} \begin{split} \label{eqn:lipschitzCondition}
| \ln\mathbb{P}(Qu\in\mathcal{S}) - \ln\mathbb{P}(Qu'\in\mathcal{S}) | \leq \epsilon \|u-u'\|, \\
 \forall u,u'\in\mathcal{U} \text{ and } \mathcal{S}\subseteq\mathcal{Y}.
\end{split} \end{align}
\end{definition}

%% PRACTICAL DEFINITION
In practical applications, the space of private data $\mathcal{U}=\mathbb{R}^{n}$ is Euclidean equipped with the $\ell_{p}$-norm. Assuming the mechanism $Q$ possesses a probability density function $g(u,y)=\mathbb{P}(Qu=y)$, where $g(u,y)$ is almost everywhere differentiable in $u$, the Lipschitz condition \eqref{eqn:lipschitzCondition} translates to a point-wise bound on the derivative across the private input $u$ as follows:
\begin{align}
g(\cdot,y) & \text{ is continuous for all } y\in\mathcal{Y} \text{ and, } \\
\| \nabla g(u,y) \|_{*} & \leq \epsilon g, \text{ for a.e. } u\in\mathcal{U} \text{ and all } y\in\mathcal{Y},
\end{align}
where $\|\cdot\|_{*}$ is the dual norm of $\|\cdot\|$.

%% METRIC IN PLACE OF ADJACENCY RELATION
\subsection{A Metric as Adjacency Relation}
The adjacency relation $\mathcal{A}$ in differential privacy is replaced by the metric $\|\cdot\|$ of the space $\mathcal{U}$ of private data. The composite adjacency relation \eqref{eqn:hybridAdjacency} can be captured using $\ell_{1}$ and $\ell_{2}$-norms. Specifically, assume that the private data $u=[u_{1},\ldots,u_{n}]$ is an aggregation of $n$ individuals' high-dimensional data $u_{i}\in\reals^{m}$. Then, adjacency relation \eqref{eqn:hybridAdjacency} can be relaxed to:
\begin{align} \label{eqn:hybridAdjacency:2}
(u,u')\in\mathcal{A} \Leftrightarrow \sum_{i=1}^{n} \| u_{i} - u_{i}' \|_{2} \leq \alpha.
\end{align}
According to the Lipschitz-privacy framework and assuming existence and differentiability of the density of the mechanism, adjacency relation \eqref{eqn:hybridAdjacency:2} translates into a bound on the derivative of the mechanism:
\begin{align} \label{eqn:hybridAdjacency:3}
\| \nabla_{u_{i}} \ln g(u,y) \|_{2} \leq \epsilon, \: \forall i\in\{1,\ldots,n\}.
\end{align}
Adjacency relation \eqref{eqn:hybridAdjacency:3} can be viewed as an $\ell_{2}$-sensitivity constraint that ensures privacy of high-dimensional data. This constraint is encapsulated in an $\ell_{1}$-sensitivity constraint that respects individuals' participation in the scheme. Additionally, this expression ensures that privacy of individuals' sensitive data remains invariant under rotation transformations on the high-dimensional data $u_{i}$. This invariance is important in many theoretical and practical case such as privacy of the state of dynamical systems and privacy of GPS traces, respectively.

%% EXAMPLES AND CONNECTIONS TO DIFF PRIVACY %%
\subsection{Connections between Lipschitz and Differential Privacy}
The notion of Lipschitz privacy is closely related to that of differential privacy. Particularly, an $\epsilon$-Lipschitz private mechanism is also differential private.
\begin{proposition} \label{thm:lipschitzToDifferentialPrivacy}
For any $\alpha>0$. Then, an $\epsilon$-Lipschitz private mechanism $Q$ is $\alpha\epsilon$-differentially private:
\begin{align}
 \prob(Qu\in\mathcal{S}) \leq e^{\epsilon} \prob(Qu'\in\mathcal{S}), \forall u,u':\: \|u-u'\|\leq\alpha.
\end{align}
\end{proposition}

Many popular differentially private mechanisms, such as the Laplace and the exponential mechanism, are also Lipschitz-differentially private. One exception that fails to satisfy Lipschitz-privacy constraints is the staircase mechanism \cite{geng12}, since the underlying noise distribution is discontinuous. Specifically, the log-probability function $\ln \prob(Qu=y)$ is discontinuous and, hence, is \textit{not} Lipschitz.

\begin{proposition}
Let $s:\mathcal{U}\times\mathcal{Y}\rightarrow\reals$ be $L$-Lipschitz in $\mathcal{U}$. Then, the mechanism $Q$ with density
\begin{align}
\prob(Qu=y|u)\propto e^{\epsilon s(u,y)}
\end{align}
is $\epsilon L$-Lipschitz differentially private. 
\end{proposition}

In the special case where $\mathcal{U}=\mathcal{Y}=\mathbb{R}^{n}$ and $s(u,y) = - \| u - y \|_{p}$, we recover the Laplace mechanism. Furthermore, Lipschitz privacy inherits the property of resiliency to post-processing. Identically to differential privacy, any further, possibly randomized, post-processing of the output carries the same privacy guarantees.
\begin{proposition}[Post-processing] \label{thm:lipschitzPostprocessing}
Consider an $\epsilon$-Lipschitz differentially private mechanism $Q:\mathcal{U} \rightarrow \Delta\left( \mathcal{Y} \right)$ and a post-processing of the output $f:\mathcal{Y} \rightarrow \mathcal{Z}$. Then, the mechanism $f\circ Q$ is $\epsilon$-Lipschitz differentially private.
\end{proposition}

%% ADVANTAGES OF INF PRIVACY %%
Propositions \ref{thm:lipschitzToDifferentialPrivacy}-\ref{thm:lipschitzPostprocessing} establish that Lipschitz-differential privacy is a stricter version of differential privacy. Lipschitz privacy has some benefits over the original framework. Firstly, the privacy constraint is simplified; the adjacency relation is now captured by the metric of the space of private data. Furthermore, Lipschitz-differential privacy enables the use of calculus tools in designing and proving properties of mechanisms. Additionally, it provides a unified privacy framework that can support richer privacy-aware applications. Privacy is now viewed as a sensitivity constraint on the mapping between private inputs and published outputs.

\section{Optimal Private Mechanisms} \label{sec:optimalityResults}
In this section, the optimality of the Laplace mechanism is proven. Specifically, we prove that the Laplace mechanism minimizes the mean-squared error among all private mechanisms that use additive and input-independent noise. Initially, the result is derived for the case of a single-dimensional identity query. Next, the result is extended to the case of isotropic multi-dimensional queries under both $\ell_{1}$ and $\ell_{2}$ norms. The $\ell_{1}$-norm respects individuals' participation in the aggregation scheme and is related to event counting queries \cite{li10}. Moreover, the $\ell_{2}$-norm is invariant under rotations and is more suitable for high-dimensional private data such as GPS signals and power consumption traces. Finally, the optimal mechanism for the case of multiple individuals contributing their high-dimensional sensitive data is derived from the results for $\ell_{1}$ and $\ell_{2}$ norms.

%% SISO IDENTITY CASE %%
\subsection{Single-Dimensional Identity Query}
The exponential mechanism introduced in \cite{mcsherry07} is a general way of building privacy-preserving mechanisms. Besides the exponential mechanism, specific mechanisms that approximate linear, high-dimensional queries were explored in \cite{li10}. However, no optimality guarantees were provided. Under the original framework of differential privacy the staircase mechanism \cite{geng12} is optimal for one-dimensional identity queries in the sense of mean-squared error. Asymptotic bounds on the sub-optimality of mechanisms approximating linear queries were introduced \cite{hardt10}. In this approach, we are interested in exact optimality results. Specifically, we provide a proof of the optimality of the Laplace mechanism for the case of single-dimensional identity queries. In \cite{wang14}, the Laplace mechanism is proven to be an entropy-minimizer. In this work, we provide a proof that the Laplace mechanism achieves the minimal mean-squared error. In the following subsections, this result is extended to high-dimensional cases.

Initially, we focus on single-dimensional private data and Lipschitz-private mechanisms that add oblivious noise. In this setting, the mean-squared error is minimized when the noise is Laplace-distributed. The problem of designing the optimal private mechanism is initially posed as an infinite-dimensional linear program. Optimality of the Laplace distribution is proven by deriving the dual problem and constructing a dual feasible solution. In particular, the space of private data is the real line $\mathcal{U}=\reals$ equipped with the absolute value as a metric. We approximate the identity query $q(u)=u$ with an $\epsilon$-Lipschitz private mechanism $Q$ that adds input-independent noise with probability measure $g$:
\begin{align}
Qu = u + V, \; \text{where} \: V\sim  g \in \Delta \left( \reals \right),
\end{align}
where $\Delta(\mathcal{Y})$ denotes the set of probability measures over the set $\mathcal{Y}$. The following result establishes the optimality of Laplace distribution.
\begin{theorem} \label{thm:laplaceOptimality1D}
Consider the set of $\epsilon$-Lipschitz private mechanisms $Q:\mathbb{R}\rightarrow \Delta\left( \mathbb{R} \right)$, $Qu = u + V$, that approximate the identity query $q:\reals\rightarrow\reals$, $q(u)=u$, where noise $V$ is input-independent and has probability distribution $g$. The Laplace mechanism that adds noise with density $l(v) = \frac{\epsilon}{2} e^{-\epsilon |v|}$ achieves the minimal mean-squared error:
\begin{align}
\mathbb{E}\left( Qu - q(u) \right)^{2} = \underset{V\sim g}{ \mathbb{E} }V^{2} \geq \underset{V\sim l}{\mathbb{E}}V^{2} = \frac{2}{\epsilon^{2}}.
\end{align}
\end{theorem}

\begin{proof}
A simplified but intuitive sketch of the proof is presented here. A full proof is presented in the Appendix. By definition, the optimal mechanism is the solution of the following optimization problem:
\begin{align} \begin{split} \label{eqn:laplaceOptimality1D:1}
\underset{g\in\Delta(\reals)}{\text{minimize}} \quad		& \underset{V\sim g}{\expe} V^{2} \\
\text{s.t.} \quad									& \text{$Q$ is $\epsilon$-Lipschitz private.}
\end{split} \end{align}
The optimization is assumed over the infinite-dimensional space of probability measures over the real line. For a simplified proof, we restrict our attention to probability measures that are continuous and almost everywhere differentiable. This assumption is removed in the technical proof. The privacy constraint is massaged:
\begin{align} \begin{split}
\text{$Q$ is } & \epsilon\text{-Lipschitz private } \Rightarrow \\ 
\left| \frac{d}{du} \ln \prob( Qu=y ) \right| & \leq \epsilon, \: \forall u,y \Leftrightarrow \\
\left| \frac{d}{du} \prob( V = y-u ) \right|  & \leq \epsilon \prob( V = y-u ), \: \forall u,y \Leftrightarrow \\
\left| g'(v) \right| & \leq \epsilon g(v), \: \forall v.
\end{split} \end{align}
Specifically, $g$ should be continuous and $g'$ should exist almost everywhere. Problem \eqref{eqn:laplaceOptimality1D:1} can, then, be restated as a linear program:
\begin{align} \begin{split} \label{eqn:laplaceOptimality1D:2}
\underset{g: AC(\reals\rightarrow\reals_{+}) }{\text{minimize}} \quad	& \int_{\reals} v^{2} g(v) dv \\
\text{s.t.} \quad													& \int_{\reals} g(v) dv = 1, \\
																		& -\epsilon g(v) \leq g'(v) \leq \epsilon g(v), \; \forall v\in\reals,
\end{split} \end{align}
where $AC$ denotes the set of absolutely continuous functions. \todo{Check the relation between a.e. diff. and abs. cont.} Problem \eqref{eqn:laplaceOptimality1D:2} is an infinite-dimensional linear program with uncountably many constraints. We assign the dual variables $\lambda\in\reals$ and $\kappa,\mu:\reals\rightarrow\reals_{+}$ for the two constraints, respectively. The dual of Problem \eqref{eqn:laplaceOptimality1D:2} is:
\begin{align} \begin{split} \label{eqn:laplaceOptimality1D:3}
\underset{\lambda\in\mathbb{R}, \eta\in\mathcal{C}^{1}\left( \reals\rightarrow\reals \right)}{\text{maximize}} \quad & \lambda \\
\text{s.t.} \quad				& \eta'(v) + \epsilon |\eta(v)| \leq v^{2} - \lambda, \: \forall v\in\mathbb{R}, \\
								& \lim_{v\rightarrow\infty} \eta(v) \geq 0, \quad \lim_{v\rightarrow-\infty} \eta(v) \leq 0.
\end{split} \end{align}
Once both primal Problem \eqref{eqn:laplaceOptimality1D:2} and dual Problem \eqref{eqn:laplaceOptimality1D:3} are stated, we construct primal and dual feasible solutions, summon weak duality, and establish optimality. The Laplace distribution $g(v) = \frac{\epsilon}{2} e^{-\epsilon |v|}$ is a primal feasible solution for Problem \eqref{eqn:laplaceOptimality1D:2} with cost $\frac{2}{\epsilon^{2}}$. Moreover, we construct a dual feasible solution for Problem \eqref{eqn:laplaceOptimality1D:3} with cost arbitrarily close to $\lambda^{*}=\frac{2}{\epsilon^{2}}$. Specifically, for any $\lambda < \lambda^{*}$, we are able to construct a dual feasible solution $\left( \lambda, \eta \right)$ that satisfies the initial value problem:
\begin{align} \label{eqn:laplaceOptimality1D:4}
\eta(0) = 0  \text{ and } \eta'(v) + \epsilon |\eta(v)| = v^{2} - \lambda, \; \forall v\in\reals\backslash\{0\}.
\end{align}
Figure \ref{laplacianOptimalityPlot} plots the unique solution $\eta:\reals\rightarrow\reals$ of the initial value problem \eqref{eqn:laplaceOptimality1D:4} for different values of $\lambda$. For $\lambda<\lambda^{*}$, the unique solution $\eta$ of the initial value problem \eqref{eqn:laplaceOptimality1D:4} is feasible since it satisfies the boundary constraints:
\begin{align}
\lim_{v\rightarrow\infty} \eta(v) \geq 0, \quad \lim_{v\rightarrow-\infty} \eta(v) \leq 0.
\end{align}
On the contrary, the dual variable $\eta$ is infeasible for $\lambda\geq\lambda^{*}$. Weak duality establishes the optimality of the Laplace mechanism. Surprisingly, the dual solution $\eta(v) = -\frac{1}{\epsilon^{2}} v (\epsilon|v|+2)$ for the optimal value $\lambda^{*}$ is infeasible. The infinite dimensionality of the problem leads to an open set of feasible solutions for problem \eqref{eqn:laplaceOptimality1D:3} and generates this paradox.

\begin{figure} \begin{center}
\includegraphics[width=.9\linewidth]{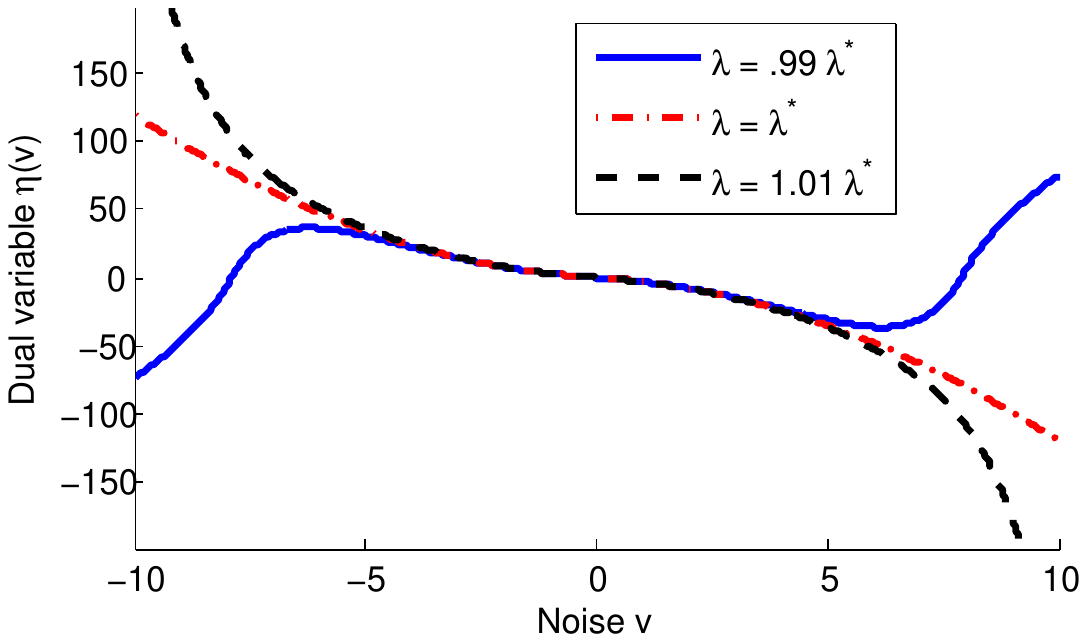}
\caption{The dual variable $\eta(v)$ is the solution to the intial value problem $\eta'(v) + \epsilon |\eta(v)| = v^{2} - \lambda$, $\eta(0)=0$ for different values of $\lambda$. A feasible solution needs to satisfy the boundary constraint $\lim_{v\rightarrow\infty} \eta(v) \geq 0$. For $\lambda<\lambda^{*}$, the solution $\eta$ is feasible.} \label{laplacianOptimalityPlot}
\end{center} \end{figure}
\end{proof}

The staircase mechanism \cite{geng12} can be viewed as an approximation of the Laplace mechanism. Although it features better mean-squared error than the Laplace mechanism, the staircase mechanism is \textit{not} $\epsilon$-Lipschitz private for any finite value of $\epsilon$. Thus, the staircase mechanism is not a feasible solution to Problem \eqref{eqn:laplaceOptimality1D:1}.

\subsection{High-Dimensional Identity Query under $\ell_{1}$-norm}
Differential privacy is mainly targeted for schemes where individuals contribute their personal data to a single database. In such schemes, the sensitive data $u$ contains each individual's private data $u_{i}$ at coordinate $i$. Here, we extend the previous results to high-dimensional identity queries. Privacy-aware approximation of identity queries can be interpreted as synthetic databases which are post-processed to answer any subsequent query. More formally, let the space of sensitive data be the real space $\mathcal{U}=\reals^{n}$ equipped with the $\ell_{1}$-norm. We focus on the case of identity queries $q:\reals^{n} \rightarrow \reals^{n}$ with $q(u)=u$. A generalized version of Theorem \ref{thm:laplaceOptimality1D} establishes optimality of the Laplace mechanism:

\begin{theorem} \label{thm:NDL1OptimalMechanism}
Consider the $\epsilon$-Lipschitz private (with respect to the $\ell_{1}$-norm) mechanism $Q:\reals^{n}\rightarrow \Delta\left( \reals^{n} \right)$ of the form $Qu = u + V$, with $V\sim g(V) \in \Delta \left( \reals^{n} \right)$. Then, the Laplace mechanism that adds oblivious noise with density $g = l^{n}_{1}(v) = \left(\frac{\epsilon}{2}\right)^{n} e^{-\epsilon \|v\|_{1}}$ minimizes mean-squared error:
\begin{align}
\underset{V\sim g}{\expe} \|V\|^{2} \geq \underset{ V \sim l^{n}_{1} }{ \expe } \|V\|_{2}^{2} = \frac{2n}{\epsilon^{2}}.
\end{align}
\end{theorem}

\begin{proof}
Similarly to the proof of Theorem \ref{thm:laplaceOptimality1D}, the optimal mechanism is the solution of the following optimization problem:
\begin{align} \begin{split} \label{eqn:NDL1OptimalMechanism:1}
\underset{ g: AC(\reals^{n}\rightarrow \reals_{+}) }{\text{minimize}} \quad	& \int_{\reals^{n}} g(v) v^{T}v dv \\
\text{s.t.} \quad															& \int_{\reals^{n}} g(v) dv = 1, \\
																					& \| \nabla g(v) \|_{\infty} \leq \epsilon g(v), \; \forall v\in\reals^{n}.	
\end{split} \end{align}
The last constraint is equivalent to
\begin{align} \label{eqn:NDL1OptimalMechanism:3}
-\epsilon g(v) \leq \frac{\partial g}{\partial v_{i}} \leq \epsilon g(v), \; \forall v\in\reals^{n}, \; \forall i\in\{1,\ldots,n\}.
\end{align}
We consider the dual variables $\lambda\in\reals$ and $\kappa_{i},\mu_{i}:\reals^{n}\rightarrow\reals_{+}$, set $\eta_{i}(v) = \mu_{i}(v) - \kappa_{i}(v)$, and derive the dual problem:
\begin{align} \begin{split} \label{eqn:NDL1OptimalMechanism:2}
\underset{ \lambda\in\reals, \eta_{i}\in\mathcal{C}^{1}\left( \reals^{n}\rightarrow\reals \right) }{\text{maximize}} \quad		& \lambda	\\
\text{s.t.} \quad										& \sum_{i=1}^{n} \left\{ \frac{\partial \eta_{i}}{\partial v_{i}} + \epsilon |\eta_{i}(v)| \right\} \leq \sum_{i=1}^{n} v_{i}^{2} - \lambda, \\
																& \lim_{v_{i}\rightarrow\infty} \eta_{i}(v) \geq 0, \; \lim_{v_{i}\rightarrow-\infty} \eta_{i}(v) \leq 0 , \; \forall i.
\end{split} \end{align}
The solution $g(v) = \left(\frac{\epsilon}{2}\right)^{n} e^{-\epsilon \|v\|_{1}}$ is feasible for the primal Problem \eqref{eqn:NDL1OptimalMechanism:1} and features cost $\frac{2n}{\epsilon^{2}}$. A feasible solution for the dual Problem \eqref{eqn:NDL1OptimalMechanism:2} is defined as:
\begin{align}
\eta_{i}(v) = \eta_{1\text{D}}(v_{i}), \quad \lambda = n \lambda_{1\text{D}},
\end{align}
where $(\lambda_{1\text{D}}, \eta_{1\text{D}})$ is a feasible dual solution for the single-dimensional case given by the initial value problem \eqref{eqn:laplaceOptimality1D:4}. Therefore, the dual Problem \eqref{eqn:NDL1OptimalMechanism:2} admits a feasible solution with cost arbitrarily close to $\frac{2n}{\epsilon^{2}}$. Weak duality establishes the optimality of the Laplace mechanism.
\end{proof}

\subsection{High-Dimension Identity Query under $\ell_{2}$-norm}
Differential privacy with respect to the $\ell_{1}$-norm captures privacy against the participation of individual users. The $\ell_{2}$-norm is a more suitable for users that contribute high-dimensional data such as GPS and power consumption traces. Once again, a version of the Laplace mechanism is proven to achieve minimum mean-squared-error among all $\epsilon$-Lipschitz private mechanisms that approximate the identity query by adding oblivious noise:

\begin{theorem} \label{thm:NDL2OptimalMechanism}
Consider the $\epsilon$-Lipschitz private (with respect to the $\ell_{2}$-norm) mechanism $Q:\reals^{n}\rightarrow \Delta\left(\reals^{n}\right)$ of the form $Qu = u + V$, with $V\sim g \in \Delta \left( \reals^{n} \right)$. Then, the Laplace mechanism that adds noise $V$ with density $g = l^{n}_{2}(v) \propto e^{-\epsilon \|v\|_{2}}$ minimizes the mean-squared error:
\begin{align}
\underset{V\sim g}{\expe}\|V\|^{2} \geq \underset{ V\sim l^{n}_{2} }{\expe}\|V\|_{2}^{2} = \frac{n(n+1)}{\epsilon^{2}}.
\end{align}
\end{theorem}

\begin{proof}
Once again, the optimal private mechanism is posed as an optimization problem:
\begin{align} \begin{split} \label{eqn:NDL2OptimalMechanism:1}
\underset{g : AC\left( \reals^{n}\rightarrow\reals_{+} \right) }{\text{minimize}} \quad		& \int_{\reals^{n}} g(v) \: v^{T}v \: d^{n}v \\
\text{s.t.} \quad														& \int_{\reals^{n}} g(v) d^{n}v = 1, \\
																				& \nabla g(v) \cdot \hat{a} \leq \epsilon g(v), \: \text{for a.e. } v\in\reals^{n}, \\
																				& \quad \forall \hat{a}\in \reals^{n}, \|\hat{a}\|_{2}=1,
\end{split} \end{align}
where the last constraint is equivalent to the privacy constraint $\|\nabla g(v)\|_{2}^{*}\leq \epsilon g(v)$. Consider the dual variables $\lambda\in\reals$ and $\kappa: \reals^{n} \times \mathbb{S}^{n-1}\rightarrow\reals_{+}$, where $\mathbb{S}^{n-1} = \{\hat{a}\in\reals^{n} : \: \|\hat{a}\|_{2}=1 \}$. Moreover, set $\eta(v) = \kappa(v)-\mu(v)$, and formulate the dual problem of Problem (\ref{eqn:NDL2OptimalMechanism:1}):
\begin{align} \begin{split} \label{eqn:NDL2OptimalMechanism:2}
\underset{\lambda\in\reals, \kappa\in \reals^{n}\times\mathbb{S}^{n-1}\rightarrow\reals_{+}}{\text{maximize}} \quad		& \lambda	\\
\text{s.t.} \quad									& \nabla \cdot \left( \int_{\mathcal{S}^{n}} \hat{a} \kappa(v,\hat{a}) d\hat{a} \right) \\
& \quad + \epsilon \int_{\mathcal{S}^{n}} \kappa(v,\hat{a}) d\hat{a}  \leq v^{T}v -\lambda, \\
															& \lim_{\|v\|_{2}\rightarrow\infty} \int_{\mathcal{S}^{n}} \hat{a}\cdot v \: \kappa(v,\hat{a}) d\hat{a} \geq 0.
\end{split} \end{align}
A feasible solution for the primal problem (\ref{eqn:NDL2OptimalMechanism:1}) is:
\begin{align} \label{eqn:NDL2OptimalMechanism:5}
g(v) = \frac{ \epsilon^{n} \Gamma\left(\frac{n}{2}+1\right) }{ \pi^{\frac{n}{2}} \Gamma(n+1) } e^{-\epsilon \|v\|_{2}},
\end{align}
with mean-squared error $\lambda^{*}=\frac{n(n+1)}{\epsilon^{2}}$. On the other hand, there exists a dual feasible solution for Problem \eqref{eqn:NDL2OptimalMechanism:2} with cost arbitrarily close to $\lambda^{*}$. Consider a dual feasible solution of the form:
\begin{align} \begin{split}
\kappa(v,\hat{a}) 	& = \left[ \eta(\|v\|_{2}) \right]^{+} \delta\left( \hat{a} + \frac{v}{\|v\|_{2}} \right) \\
					& \quad + \left[ \eta(\|v\|_{2}) \right]^{-} \delta\left( \hat{a} - \frac{v}{\|v\|_{2}} \right),
\end{split} \end{align}
where $\delta$ is Dirac's delta function on the unit $n$-sphere $\mathbb{S}^{n-1}$, $\eta:\reals_{+}\rightarrow\reals$ is a suitable function, and $[\cdot]^{+}$ and $[\cdot]^{-}$ are the positive and negative parts of a function, respectively. Then, we can reduce the feasible region of Problem \eqref{eqn:NDL2OptimalMechanism:2} and rewrite it as
\begin{align} \begin{split} \label{eqn:NDL2OptimalMechanism:3}
\underset{ \lambda\in\reals, \eta: \reals_{+}\rightarrow\reals }{\text{maximize}} \quad	& \lambda \\
\text{s.t.} \quad																			& \eta'(r) + \frac{n-1}{r} \eta(r) + \epsilon |\eta(r)| \leq r^{2} - \lambda \\
																									& \lim_{r\rightarrow\infty} \eta(r) \geq 0.
\end{split} \end{align}
Similarly to the proof of Theorem \ref{thm:NDL1OptimalMechanism}, a feasible solution $(\lambda,\eta)$ of Problem \eqref{eqn:NDL2OptimalMechanism:3} of the following form is constructed:
\begin{align} \begin{split} \label{eqn:NDL2OptimalMechanism:4}
 \eta'(r) + \frac{n-1}{r} \eta(r) + \epsilon |\eta(r)| = r^{2} - \lambda \text{ and } \eta(0)=0
\end{split} \end{align}
Figure \ref{fig:typicalDualSolutionNDL2} shows the solution of the initial value problem \eqref{eqn:NDL2OptimalMechanism:4} for different values of $\lambda$. For $\lambda < \lambda^{*}$, the solution is feasible and, thus, the optimality of the density (\ref{eqn:NDL2OptimalMechanism:5}) for the initial value problem \eqref{eqn:NDL2OptimalMechanism:1} is established.

Again, for $\lambda=\lambda^{*}$, the dual solution $\eta(r) = - \frac{r (r\epsilon + n + 1)}{\epsilon^{2}} $ is infeasible as a result of the infinite-dimensional nature of problem (\ref{eqn:NDL2OptimalMechanism:4}).
\end{proof}

\begin{figure} \begin{center}
\includegraphics[width=.9\linewidth]{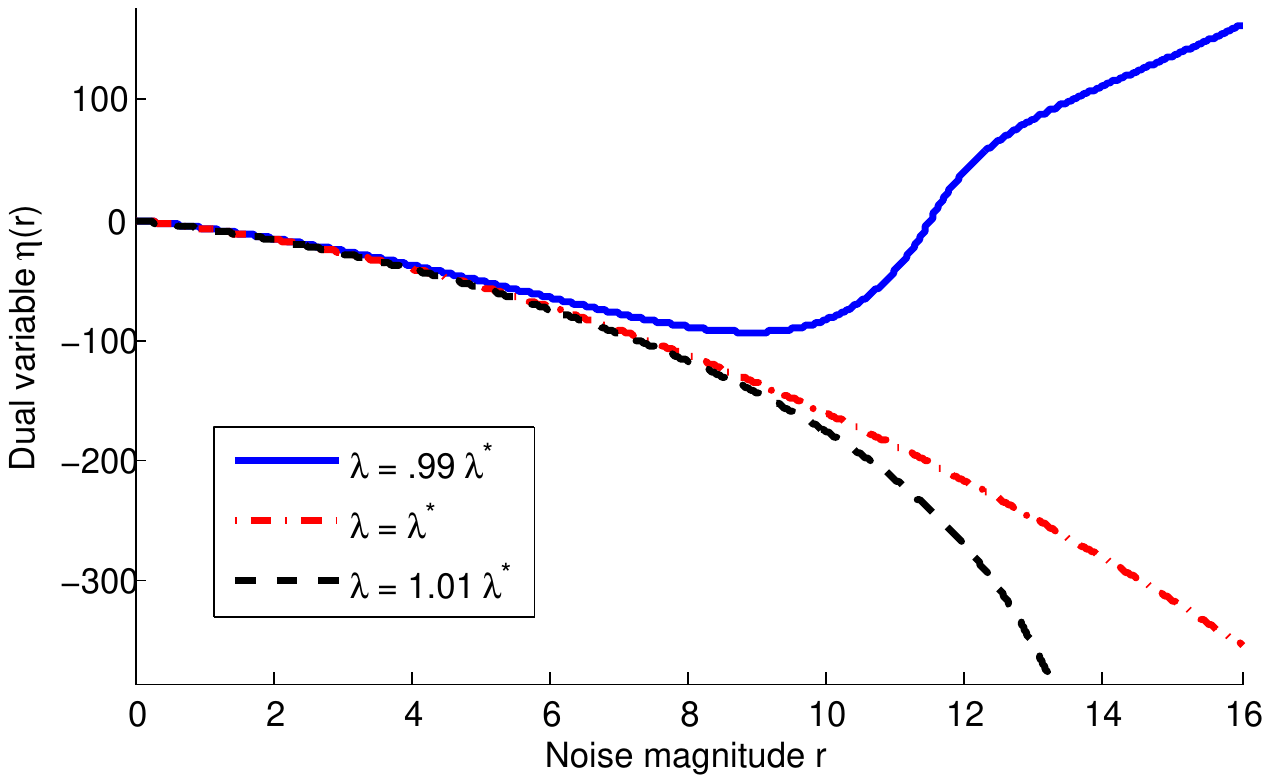}
\caption{The dual variable $\eta(v)$ is the solution to the intial value problem $\eta'(r) + \frac{n-1}{r} \eta(r) + \epsilon |\eta(r)| = r^{2} - \lambda$, $\eta(0)=0$ for different values of $\lambda$. A feasible solution needs to satisfy the boundary constraint $\lim_{v\rightarrow\infty} \eta(v) \geq 0$. For $\lambda<\lambda^{*}$, the solution $\eta$ is feasible.} \label{fig:typicalDualSolutionNDL2}
\end{center} \end{figure}

Sample from distribution \eqref{eqn:NDL2OptimalMechanism:5} can be efficiently generated. The magnitude $r=\|v\|_{2}$ of the noise is drawn from the Gamma distribution $r\sim \frac{\epsilon^{n}}{\Gamma(n)} e^{-\epsilon r} r^{n-1}$ and the direction $\hat{v}=\frac{v}{\|v\|_{2}}$ is uniformly sampled from the sphere $\mathbb{S}^{n-1}$.

\subsection{Multiple Users with High-Dimensional Private Data}
In this section, the case of multiple users contributing their high-dimensional sensitive data is explored. Specifically, consider $n$ individuals. Each individual contributes his $m$-dimensional sensitive data $u_{i}\in\reals^{n}, \: n\in\{1,\ldots,n\}$. Furthermore, we are interested in releasing a privacy-aware version of the sensitive data under an adjacency relation that preserves both individual's participation and each user's data. These aspects of privacy are captured by adjacency relation (\ref{eqn:hybridAdjacency:2}) derived earlier.

In particular, let the space of private data be $\mathcal{U}=\reals^{n\cdot m}$ and consider private mechanisms $Q$ that add input-independent noise $V\sim g$ to the private data $u$. Similarly to the previous case, a version of the Laplace mechanism provides the optimal mean-squared error.

\begin{theorem} \label{thm:NDL1L2OptimalMechanism}
Consider the $\epsilon$-Lipschitz private (with respect to the adjacency relation (\ref{eqn:hybridAdjacency:2})) mechanism $Q:\reals^{n \cdot m}\rightarrow \Delta\left( \reals^{n \cdot m} \right)$ of the form $Qu = u + V$, with $V\sim g \in \Delta \left( \reals^{n\cdot m} \right)$. Then, the Laplace mechanism that adds oblivious noise with density $g = l^{n,m}(v) \propto e^{-\epsilon \sum_{i=1}^{n} \|v_{i}\|_{2}}$ minimizes the mean-squared-error:
\begin{align}
\underset{V\sim g}{\expe}\|V\|^{2} \geq \underset{ V\sim l^{n,m} }{\expe} \|V\|_{2}^{2} = \frac{nm(m+1)}{\epsilon^{2}}.
\end{align}
\end{theorem}

\begin{proof}
The primal optimization problems is as follows
\begin{align} \begin{split}
\underset{ g: AC(\reals^{n\cdot m}\rightarrow \reals_{+}) }{\text{minimize}} \quad	& \int_{\reals^{n\cdot m}} g(v) v^{T}v dv \\
\text{s.t.} \quad															& \int_{\reals^{n\cdot m}} g(v) dv = 1, \\
																					& \| \nabla_{i} g(v) \|_{2} \leq \epsilon g(v), \; \forall i\in[n], \forall v\in\reals^{n},													
\end{split} \end{align}
where $\nabla_{i} g = \begin{bmatrix} \frac{\partial g}{\partial v_{(i-1)\cdot m + 1}} & \ldots & \frac{\partial g}{\partial v_{i\cdot m}}\end{bmatrix}$ and $[n]=\{1,\ldots,n\}$. The dual problem is formulated:
\begin{align} \begin{split} \label{eqn:NDL1L2OptimalMechanism:2}
\underset{ \lambda\in\reals, \eta_{i}: \reals_{+}\rightarrow\reals }{\text{maximize}} \quad		& \lambda	\\
\text{s.t.} \quad																				& \sum_{i=1}^{n} \left\{ \eta_{i}'(r_{i}) + \frac{n-1}{r_{i}} \eta_{i}(r_{i}) + \epsilon |\eta_{i}(r_{i})| \right\}  \\	
																								& \leq \sum_{i=1}^{n} r_{i}^{2} - \lambda, \text{ and } \lim_{r_{i}\rightarrow\infty} \eta_{i}(r_{i}) \geq 0, \; \forall i 
\end{split} \end{align}
A pair of feasible primal and dual solutions is constructed:
\begin{align} \begin{split}
& g = \left(\frac{ \epsilon^{m} \Gamma(\frac{m}{2}+1) }{ m\pi^{\frac{m}{2}}\Gamma(m) }\right)^{n} e^{-\epsilon \sum_{i=1}^{n} \|v_{i}\|_{2}},	\\
& \eta_{i}(r_{i}) = \eta_{\ell_{2}}(r_{i}), \text{ and } \lambda = n \lambda_{\ell_{2}},
\end{split} \end{align}
where $(\lambda_{\ell_{2}}, \eta_{\ell_{2}})$ is the dual solution of Theorem \ref{thm:NDL2OptimalMechanism}. Weak duality establishes the optimality of the solution. 
\end{proof}

\section{Discussion} \label{sec:discussion}
In this work, we explored Lipschitz privacy, which is a version of differential privacy that is adapted for metric spaces. Moreover, we proved that, for a given privacy level, the Laplace mechanism minimizes the mean-squared error among all single-dimensional mechanisms that add input-independent noise. The design of the optimal private mechanism is initially formulated as a linear program. Then, the optimality of the Laplace mechanism is established by constructing a pair of primal and dual feasible solutions with zero duality gap. Next, the result is extended to high-dimensional real spaces equipped with the $\ell_{1}$-norm. The case of $\ell_{1}$-norm corresponds to the case of providing privacy guarantees with respect to participation of any individual. Furthermore, the optimality of a variation of the Laplace mechanism is established for real spaces equipped with the $\ell_{2}$-norm. In this case, the privacy guarantees are invariant under rotations and, thus, this choice of norm captures the case where every individual provides high-dimensional sensitive data. A combination of the two results provides the optimal privacy-aware approximation of the aggregation of high-dimensional sensitive data of multiple individuals. Future directions include optimality guarantees for more general classes of queries beyond identity queries. Moreover, it is useful to study optimality results for other composite adjacency relations such as that proposed in \cite{koufogiannis14}.

\bibliographystyle{unsrt}
\bibliography{laplaceOptimality}

\vspace{-15pt}
\appendix

\vspace{-4pt}
\section{Extended Proofs} \label{app:proofs}
Theorem \ref{thm:laplaceOptimality1D} establishes the optimality of the Laplace mechanism for a single-dimensional identity query. A more technical proof is presented here. First, we prove that, for Lipschitz differential privacy guarantees to hold, the additive noise should possess density.

\begin{lemma} \label{lem:densityExists}
Consider the $\epsilon$-Lipschitz private mechanism $Q$ that uses oblivious, additive noise $V$. Specifically, let $Qu= u+V$, where $V$ has probability measure $g\in\Delta\left( V \right)$. Then $V$ possesses density.
\end{lemma}
\begin{proof}
We prove that the cumulative density function $G$ of $V$
\begin{align}
G(x) = \prob( V\leq x)
\end{align}
is absolutely continuous. For any measurable $S\subseteq\reals$ and any $u_{1},u_{2}\in\reals$, Lipschitz privacy dictates that:
\begin{align}
\left| \ln \prob(Qu_{1}\in S) - \ln \prob(Qu_{2}\in S) \right| \leq \epsilon | u_{1} - u_{2} |
\end{align}

Let $S=(-\infty,0]$, $u_{1}=-x$, and $u_{2}=-y$, with $x<y$. Then:
\begin{align}
\left| \ln \prob(V\leq x) - \ln \prob(V\leq y) \right| & \leq \epsilon |x-y| \Rightarrow \\
\left| \prob(V\leq x) - \prob(V\leq y) \right| & \leq \prob(V\leq x) \epsilon |x-y| \Rightarrow \\
\left| G(x) - G(y) \right| & \leq \epsilon |x-y|
\end{align}
Therefore, $G$ is absolutely continuous and, hence, $V$ possesses density. Abusing notation, we denote the density of the noise $V$ with $g$.
\end{proof}

We now provide a technical proof of Theorem \ref{thm:laplaceOptimality1D}.
\begin{proof}
Consider the $\epsilon$-Lipschitz differential private mechanisms that use additive, oblivious noise $V$ with probability measure $g$:
\begin{align}
Q:\reals\rightarrow\Delta(\reals), \quad Qu = u + V, \text{ where } V\sim g.
\end{align}
Solving for the optimal, in the mean-squared error sense, probability measure is posed as a linear, but infinite-dimensional program:
\begin{align} \label{thm.prob1:laplaceOptimality1D} \begin{split}
\underset{g\in\Delta(\reals)}{ \text{minimize} } \quad	& \underset{ V\sim g }{\expe} V^{2} \\
						\text{s.t.} 		\quad & \text{$g$ is $\epsilon$-Lipschitz diff. private}
\end{split} \end{align}
Lemma \ref{lem:densityExists} establishes that $V$ possesses density which is abusively denoted by $g(v)$. Therefore, Problem \eqref{thm.prob1:laplaceOptimality1D} is equivalently written as:
\begin{align} \label{thm.prob2:laplaceOptimality1D} \begin{split} 
\underset{g:\mathcal{C}^{1}\left( \mathbb{R}\rightarrow\mathbb{R} \right)}{\text{minimize}} \quad & \int_{\mathbb{R}} g(v) v^{2} dv \\
\text{s.t.} \quad 							& \int_{\mathbb{R}} g(v) dv = 1, \text{ and } g(v) \geq 0, \: \forall v, \\
											& -\epsilon g(v) \leq \liminf_{\delta\rightarrow0} \frac{g(v+\delta)-g(v)}{\delta}, \: \forall v \\
											& \limsup_{\delta\rightarrow0} \frac{g(v+\delta)-g(v)}{\delta} \leq \epsilon g(v), \: \forall v.
\end{split} \end{align}
Problem \eqref{thm.prob2:laplaceOptimality1D} is an infinite-dimensional linear program with infinite many constraints, thus, it is unclear though whether the minimum is achievable. The Laplace distribution $l_{\epsilon}(v) = \frac{\epsilon}{2} e^{-\epsilon |v|}$ is a feasible solution with mean error $\frac{2}{\epsilon^{2}}$. We now discritize, dualize and take limits in order to compute the dual problem. As a result, we prove that the dual variable is differentiable and we retrieve the formulation of the dual problem. Consider $N$ discrete points:
\begin{align}
v_{i} = - M + i \cdot \nu, \quad i\in\{1,\ldots,N\}
\end{align}
where $\nu = \frac{2M}{N-1}$ is the discritization step and $M$ is the truncation limit. For $g_{i} = g(v_{i})$, the original optimization is problem is now approximated by its discritized version:
\begin{align} \begin{split} 
\underset{\{g_{i}\}_{i=1}^{N}\in\reals^{N}}{\text{minimize}} \; & \sum_{i=1}^{N} g_{i} v_{i}^{2} \nu \\
\text{s.t.} \; 							& \sum_{i=1}^{N} g_{i} \nu = 1, \text{ and } g_{i} \geq 0, \; \forall i, \\
											& -\epsilon \cdot \frac{g_{i} + g_{i+1}}{2} \leq \frac{g_{i+1} - g_{i}}{\nu} \leq \epsilon \cdot \frac{g_{i} + g_{i+1}}{2}, \: \forall i. \\
\end{split} \end{align}
Let $\lambda\in\reals$, and $\kappa_{i},\mu_{i}\in\reals_{+}$ with $i\in\{1,\ldots,N-1\}$ be the dual variables for the first and the second constraint, respectively. The Lagrangian of the optimization problem is computed and minimized over $\{g_{i}\}_{i=1}^{N}\in\reals^{N}_{+}$:
\begin{align}
\mathcal{L}(g,\lambda,\kappa,\mu)  = 	& \sum_{i}^{N} g_{i} v_{i}^{2} \nu + \lambda - \lambda \sum_{i=1}^{N} g_{i} \nu \\
										& + \sum_{i=1}^{N-1} \left[ -\epsilon \kappa_{i} \frac{g_{i}+g_{i+1}}{2} - \kappa_{i} \frac{g_{i+1}-g_{i}}{\nu} \right] \\
										& + \sum_{i=1}^{N-1} \left[ \mu_{i} \frac{g_{i+1}-g_{i}}{\nu} - \epsilon \mu_{i} \frac{g_{i}+g_{i+1}}{2}  \right]
\end{align}
%Equivalently:
%\begin{align}
%& \mathcal{L}(g,\lambda,\kappa,\mu)  = \lambda + \sum_{i=2}^{N-1} g_{i} \left( v_{i}^{2} \nu - \lambda \nu - \epsilon \frac{\kappa_{i-1}+\kappa_{i}}{2} \right. \\
%										& \left. - \epsilon \frac{\mu_{i-1}+\mu_{i}}{2} + \frac{\kappa_{i}-\kappa_{i-1}}{\nu} - \frac{\mu_{i}-\mu_{i-1}}{\nu} \right) \\
%										& + g_{1} \left( v_{1}^{2} \nu - \lambda \nu - \epsilon \frac{\kappa_{1}+\mu_{1}}{2} - \frac{\mu_{1}-\kappa_{1}}{\nu} \right) \\
%										& + g_{N} \left( v_{N}^{2} \nu - \lambda \nu - \epsilon \frac{\kappa_{N-1}+\mu_{N-1}}{2} + \frac{\mu_{N-1}-\kappa_{N-1}}{\nu}  \right)
%\end{align}
Thus, the dual problem is the following:
\begin{align}
\underset{\lambda,\{\kappa_{i}\},\{\mu_{i}\}}{\text{maximize}} \; & \lambda \\
\text{s.t.} \;	&  \epsilon \frac{\kappa_{i-1}+\kappa_{i}}{2} + \epsilon \frac{\mu_{i-1}+\mu_{i}}{2} - \frac{\kappa_{i}-\kappa_{i-1}}{\nu} \\
							& + \frac{\mu_{i}-\mu_{i-1}}{\nu} \leq v_{i}^{2} \nu - \lambda \nu, \: \forall i\in\{2,\ldots,N-1\}, \\
							& v_{1}^{2} \nu - \lambda \nu - \epsilon \frac{\kappa_{1}+\mu_{1}}{2} - \frac{\mu_{1}-\kappa_{1}}{\nu} \geq 0, \\
							& v_{N}^{2} \nu - \lambda \nu - \epsilon \frac{\kappa_{N-1}+\mu_{N-1}}{2} \\
							& + \frac{\mu_{N-1}-\kappa_{N-1}}{\nu} \geq 0, \\
							& \kappa_{i} \geq 0 \text{ and } \mu_{i} \geq 0, \quad \forall i\in\{1,\ldots,N\}
\end{align}
Complementary slackness of the primal problem suggests that, for each $i$, either $\kappa_{i}=0$ or $\mu_{i}=0$. Therefore, we seek dual feasible solutions such that $\eta_{i} = \mu_{i}-\kappa_{i}$  and $|\eta_{i}| = \mu_{i} + \kappa_{i}$:
\begin{align} \begin{split}
\underset{\lambda,\{\eta_{i}\}}{\text{maximize}} \quad & \lambda \\
\text{s.t.} \quad	&  \epsilon \frac{|\eta_{i-1}|+|\eta_{i}|}{2} + \frac{\eta_{i}-\eta_{i-1}}{\nu} \leq v_{i}^{2} \nu - \lambda \nu, \\
							& \quad \forall i\in\{2,\ldots,N-1\}, \\
							& v_{1}^{2} \nu - \lambda \nu - \epsilon \frac{|\eta_{1}|}{2} - \frac{\eta_{1}}{\nu} \geq 0, \\
							& v_{N}^{2} \nu - \lambda \nu - \epsilon \frac{|\eta_{N-1}|}{2} + \frac{\eta_{N-1}}{\nu} \geq 0
\end{split} \end{align}
We first set $N = \frac{2M}{\nu} + 1$ and let $M\rightarrow\infty$ and, then, let $\nu\rightarrow0$. The discritized dual problem convergences to the continuous one:
\begin{align} \begin{split}
\underset{\lambda\in\mathbb{R}, \eta:\mathcal{C}^{1}\left( \mathbb{R}\rightarrow\mathbb{R} \right)}{\text{maximize}} \quad & \lambda \\
\text{s.t.} \quad				& \eta'(v) + \epsilon |\eta(v)| \leq v^{2} - \lambda, \: \forall v\in\mathbb{R}, \\
								& \lim_{v\rightarrow\infty} \eta(v) \geq 0, \quad \lim_{v\rightarrow-\infty} \eta(v) \leq 0
\end{split} \end{align}
The last step of the proof includes building a feasible dual solution for $\lambda = \frac{2-\delta}{\epsilon^{2}}$, for small, positive values of $\delta$. Specifically, we fix $\lambda = \frac{2-\delta}{\epsilon^{2}}$ and solve the initial value problem:
\begin{align} \label{eqn:1DLaplaceOptimalityFullProof:eqn1}
\eta'(v) + \epsilon |\eta(v)| = v^{2} - \lambda, \text{ and } \eta(0)=0
\end{align}
Existence and uniqueness of solutions for the initial value problem \eqref{eqn:1DLaplaceOptimalityFullProof:eqn1} implies that the unique solution only needs to be checked that it satisfy the boundary constraints. Some technical analysis proves that, for small and positive values of $\delta$, the solution $\eta$ to the initial value problem \eqref{eqn:1DLaplaceOptimalityFullProof:eqn1} indeed satisfies the constraints:
\begin{align}
\lim_{v\rightarrow\infty} \eta(v) \geq 0 \text{ and } \lim_{v\rightarrow-\infty} \eta(v) \leq 0
\end{align}
Due to symmetry, we focus only on the case of $v\geq0$. Table \ref{tbl:1DLaplaceOptimalityFullProof} summarizes the signs of $\eta$ and its derivatives. Specifically, the solution $\eta$ is negative until $\hat{v}_{3}$. While $\eta$ remains negative, it satisfies the initial value problem:
\begin{align}
\eta'(v) - \epsilon \eta(v) = v^{2} - \lambda \text{ and } \quad \eta(0)=0 
\end{align}
 The single root of the second derivative is analytically computed:
\begin{align}
\hat{v}_{1} = \frac{ \ln2 - \ln\delta }{\epsilon}
\end{align}
At $\hat{v}_{3}$, the dual function $\eta$ becomes positive and satisfies the initial value problem \eqref{eqn:1DLaplaceOptimalityFullProof:eqn2}:
\begin{align} \label{eqn:1DLaplaceOptimalityFullProof:eqn2}
\eta'(v) + \epsilon \eta(v) = v^{2} - \lambda \text{ and } \eta(\hat{v}_{3}) = 0 
\end{align}
The value $\hat{v}_{3}$ can become arbitrarily large. Indeed, it holds that $\hat{v}_{3}\geq \hat{v}_{1}$ and, for small enough values of $\delta$, $\hat{v}_{1}$ can become as large as needed. Therefore, for small enough values of $\delta$, the derivative of Equation \eqref{eqn:1DLaplaceOptimalityFullProof:eqn2} remains positive:
\begin{align}
\eta'(v) = \frac{2(v\epsilon - 1)}{\epsilon^{2}} + \frac{e^{(v_{3}-v)\epsilon} ( \delta - 2v_{3} \epsilon + v_{3}^{2} \epsilon^{2} )}{\epsilon^{2}} \geq 0,
\end{align}
for $v\geq \hat{v}_{3}$.

\begin{figure} \begin{center}
\includegraphics[width=.8\linewidth]{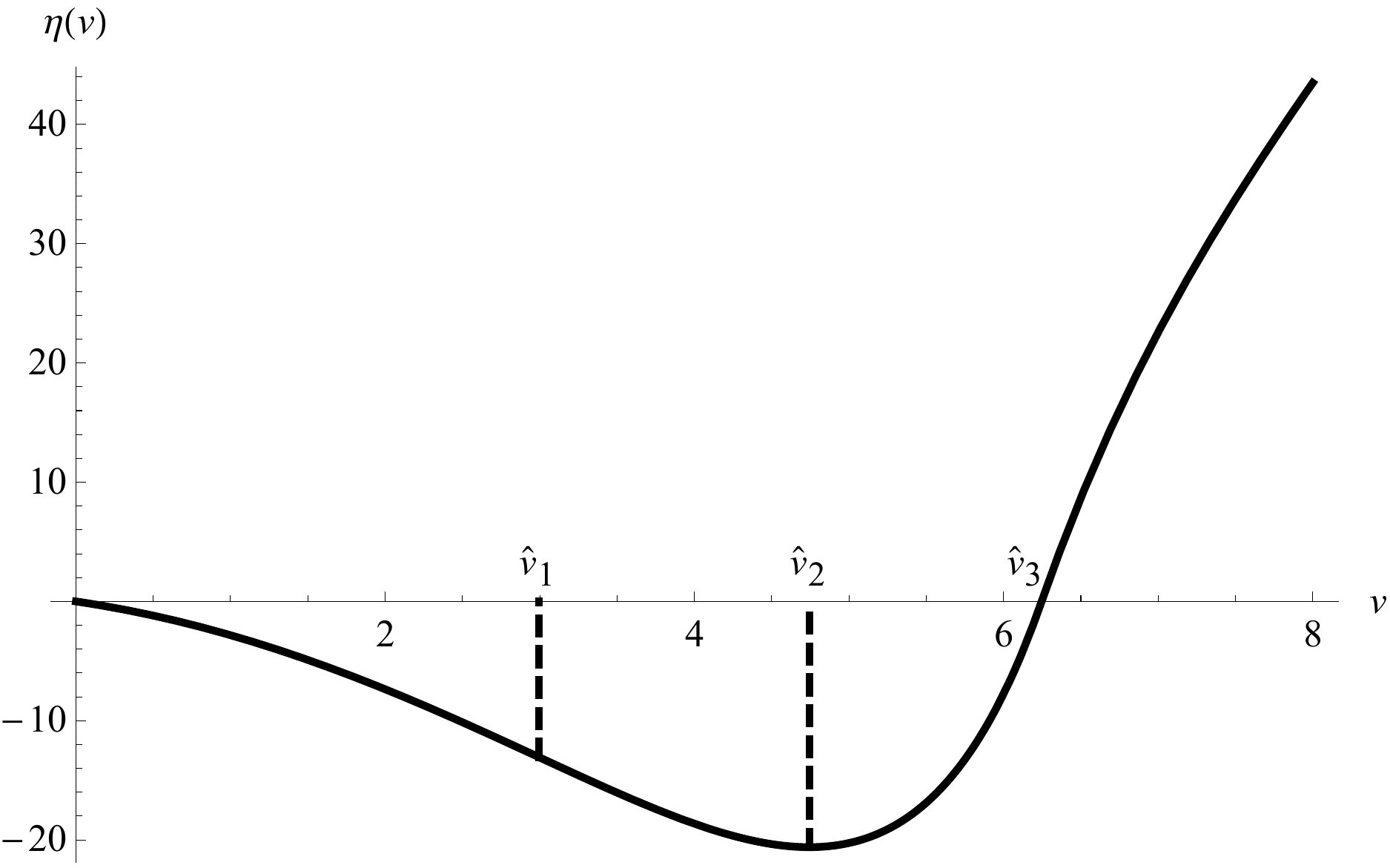}
\caption{The dual solution $\eta$ for small values of $\delta$. The function $\eta$ changes curvature at $\hat{v}_{1}$, becomes increasing at $\hat{v}_{2}$, and is zero at  $\hat{v}_{3}$. For small values of $\delta$, once $\eta$ becomes positive, it remains increasing and, thus, positive.} \label{fig:typicalDualSolution}
\end{center} \end{figure}

\begin{table} \begin{center} \begin{tabularx}{.75\linewidth}{l||cccccccc}
 \toprule
    $v$			&	$0$		&		&	$v_{1}$		&		&	$v_{2}$		&		&	$v_{3}$		&		 		\\ \midrule
    $\eta''$	& 	$-$		&	$-$	&  $0$  		&	$+$ &   $+$			&  $+$	&   $+$   		& $+$ 	    	\\
    $\eta'$		& 	$-$		&	$-$	&  $-$  		&	$-$ &   $0$			&  $+$	&   $+$   		& $+$ 	       	\\
    $\eta$		& 	$0$		&	$-$	&  $-$  		&	$-$ &   $-$			&  $-$	&   $0$   		& $+$ 	      	\\ \bottomrule
\end{tabularx}
\caption{Elementary analysis on the behaviour of function $\eta(v)$ for small and positive values of $\delta$.} 
\label{tbl:1DLaplaceOptimalityFullProof}
 \end{center} \end{table}

The cost of the constructed dual feasible solution is $\lambda=\frac{2-\delta}{\epsilon^{2}}$ and can be made as close to the cost of the Laplace distribution. Weak duality completes the proof.
\end{proof}

\begin{remark}
For $\lambda = \frac{2}{\epsilon^{2}}$, we consider the dual solution $\eta(v)=-\frac{1}{\epsilon^{2}}v(\epsilon|v|+2)$ which satisfies the differential equation. However, it fails to satisfy the boundary conditions since it quadratically explodes. Instead, for $\lambda > \frac{2}{\epsilon^{2}}$, the dual feasible explodes exponentially. Despite the qualitative difference between the two cases, they are both infeasible.
\end{remark}

\end{document}